\documentclass[onecolumn]{article}

\usepackage{graphicx,epsfig}
\usepackage{times}
\usepackage{hyperref}
\usepackage{amssymb, amsthm}
\usepackage {amssymb,latexsym,epic,eepic,stmaryrd,url}
\usepackage{bbm}
\newtheorem{thm}{Theorem}
\newtheorem{cor}{Corollary}
\newtheorem{clm}{Claim}
\newtheorem{fact}{Fact}
\newtheorem{lem}{Lemma}
\newtheorem{asm}{Assumption}
\newtheorem{obs}{Observation}

\begin{document}
\title{\Large \bf Acyclic Edge Coloring of Triangle Free Planar Graphs}

\author{Manu Basavaraju\thanks{\textbf{Corresponding Author:} Computer Science and Automation department,
Indian Institute of Science,
Bangalore- 560012,
India.  {\tt manu@csa.iisc.ernet.in; iammanu@gmail.com}} \and L. Sunil Chandran\thanks{Computer Science and Automation department,
Indian Institute of Science,
Bangalore- 560012,
India.  {\tt sunil@csa.iisc.ernet.in}}
}

\date{}
\pagestyle{plain}
\maketitle

\begin{abstract}

An $acyclic$ edge coloring of a graph is a proper edge coloring such that there are no bichromatic cycles. The \emph{acyclic chromatic index} of a graph is the minimum number k such that there is an acyclic edge coloring using k colors and is denoted by $a'(G)$.  It was conjectured by Alon, Sudakov and Zaks (and much earlier by Fiamcik) that $a'(G)\le \Delta+2$, where $\Delta =\Delta(G)$ denotes the maximum degree of the graph.

If every induced subgraph $H$ of $G$ satisfies the condition $\vert E(H) \vert \le 2\vert V(H) \vert -1$, we say that the graph $G$ satisfies $Property\ A$. In this paper, we prove that if $G$ satisfies $Property\ A$, then $a'(G)\le \Delta + 3$. Triangle free planar graphs satisfy $Property\ A$. We infer that $a'(G)\le \Delta + 3$, if $G$ is a triangle free planar graph. Another class of graph which satisfies $Property\ A$ is  2-fold graphs (union of two forests).

\end{abstract}


\section{Introduction}

All graphs considered in this paper are finite and simple. A proper \emph{edge coloring} of $G=(V,E)$ is a map $c: E\rightarrow C$ (where $C$ is the set of available $colors$ ) with $c(e) \neq c(f)$ for any adjacent edges $e$,$f$. The minimum number of colors needed to properly color the edges of $G$, is called the chromatic index of $G$ and is denoted by $\chi'(G)$. A proper edge coloring c is called acyclic if there are no bichromatic cycles in the graph. In other words an edge coloring is acyclic if the union of any two color classes induces a set of paths (i.e., linear forest) in $G$. The \emph{acyclic edge chromatic number} (also called \emph{acyclic chromatic index}), denoted by $a'(G)$, is the minimum number of colors required to acyclically edge color $G$. The concept of \emph{acyclic coloring} of a graph was introduced by Gr\"unbaum \cite{Grun}. The \emph{acyclic chromatic index} and its vertex analogue can be used to bound other parameters like \emph{oriented chromatic number} and \emph{star chromatic number} of a graph, both of which have many practical applications, for example, in wavelength routing in optical networks ( \cite{ART}, \cite{KSZ} ). Let $\Delta=\Delta(G)$ denote the maximum degree of a vertex in graph $G$. By Vizing's theorem, we have $\Delta \le \chi'(G) \le \Delta +1 $(see \cite{Diest} for proof). Since any acyclic edge coloring is also proper, we have $a'(G)\ge\chi'(G)\ge\Delta$. \newline

It has been conjectured by Alon, Sudakov and Zaks \cite{ASZ} (and much earlier by Fiamcik \cite{Fiam}) that $a'(G)\le\Delta+2$ for any $G$. Using probabilistic arguments Alon, McDiarmid and Reed \cite{AMR} proved that $a'(G)\le60\Delta$. The best known result up to now for arbitrary graph, is by Molloy and Reed  \cite{MolReed} who showed that $a'(G)\le16\Delta$. Muthu, Narayanan and Subramanian \cite{MNS1} proved that $a'(G)\le4.52\Delta$ for graphs $G$ of girth at least 220 (\emph{Girth} is the length of a shortest cycle in a graph).\newline

Though the best known upper bound for general case is far from the conjectured $\Delta+2$, the conjecture has been shown to be true for some special classes of graphs. Alon, Sudakov and Zaks \cite{ASZ} proved that there exists a constant $k$ such that $a'(G)\le\Delta+2$ for any graph $G$ whose girth is at least $k\Delta\log\Delta$. They also proved that $a'(G)\le\Delta+2$ for almost all $\Delta$-regular graphs. This result was improved by Ne\v set\v ril and Wormald \cite{NesWorm} who showed that for a random $\Delta$-regular graph $a'(G)\le \Delta+1$. Muthu, Narayanan and Subramanian proved the conjecture for grid-like graphs \cite{MNS2}. In fact they gave a better bound of $\Delta+1$ for these class of graphs. From Burnstein's \cite{Burn} result it follows that the conjecture is true for subcubic graphs. Skulrattankulchai \cite{Skul} gave a polynomial time algorithm to color a subcubic graph using $\Delta+2 = 5$ colors. Fiamcik \cite{Fiam1}, \cite{Fiam2} proved that every subcubic graph, except for $K_4$ and $K_{3,3}$, is acyclically edge colorable using $4$ colors.

Determining $a'(G)$ is a hard problem both from theoretical and algorithmic points of view. Even for the simple and highly structured class of complete graphs, the value of $a'(G)$ is still not determined exactly. It has also been shown by Alon and Zaks \cite{AZ} that determining whether $a'(G)\le3$ is NP-complete for an arbitrary graph $G$. The vertex version of this problem has also been extensively studied ( see \cite{Grun}, \cite{Burn}, \cite{Boro}). A generalization of the acyclic edge chromatic number has been studied: The \emph{r-acyclic edge chromatic number} $a'_r(G)$ is the minimum number of colors required to color the edges of the graph $G$ such that every cycle $C$ of $G$ has at least min\{$\vert C \vert$,$r$\} colors ( see \cite{GeRa}, \cite{GrePi}).

~~~~~~~~~

\noindent\textbf{Our Result:} The acyclic chromatic index of planar graphs has been studied previously. Fiedorowicz, Hauszczak and Narayanan \cite{AMN} gave an upper bound of $2\Delta+29$ for planar graphs. Independently Hou, Wu, GuiZhen Liu and Bin Liu \cite{JWu1} gave an upper bound of $max(2\Delta - 2,\Delta+22)$. Note that for $\Delta \ge 24$, it is equal to $2\Delta-2$. Basavaraju and Chandran \cite{MBSC6} improved the bound significantly to  $\Delta +12$.

The acyclic chromatic index of special classes of planar graphs characterized by some lower bounds on girth or the absence of short cycles have also been studied. In \cite{JWu1} an upper bound of $\Delta+2$ for planar graphs of girth at least $5$ has been proved. Fiedorowicz and Borowiecki \cite{AFMB} proved an upper bound of $\Delta+1$ for planar graphs of girth at least $6$ and an upper bound of $\Delta+15$ for planar graphs without cycles of length $4$. In \cite{AMN}, an upper bound of $\Delta+6$ for triangle free planar graphs has been proved. In this paper we improve the bound to $\Delta+3$. In fact we prove a more general theorem as described below:

\noindent $Property\ A:$ Let $G$ be a simple graph. If every induced subgraph $H$ of $G$ satisfies the condition $\vert E(H) \vert \le 2\vert V(H) \vert -1$, we say that the graph $G$ satisfies $Property\ A$. If $G$ satisfies $Property\ A$, then every subgraph of $G$ also satisfies $Property\ A$.

In this paper, we prove the following theorem:

\begin{thm}
\label{thm:thm1}
If a graph $G$ satisfies $Property\ A$, then $a'(G) \le \Delta(G) +3$.
\end{thm}

Note that triangle free planar graphs, 2-degenerate graphs, 2-fold graphs (union of two forests), etc.  are some classes of graphs which satisfy $Property\ A$.  The following corollary is obvious.

\begin{cor}
\label{cor:cor1}
If $G$ is a triangle free planar graph, then $a'(G)\le \Delta +3$.
\end{cor}

\indent Note that this is the best result known for triangle free planar graphs and 2-fold graphs. The earlier known bound for these classes of graphs was $\Delta + 6$ by \cite{AMN}. In case of 2-degenerate graphs a tight bound of $\Delta+1$ has been proved in \cite{MBSC4}.

Our proof is constructive and yields an efficient polynomial time algorithm. We have presented the proof in a non-algorithmic way. But it is easy to extract the underlying algorithm from it.

\section{Preliminaries}

Let $G=(V,E)$ be a simple, finite and connected graph of $n$ vertices and $m$ edges. Let $x \in V$. Then $N_{G}(x)$ will denote the neighbours of $x$ in $G$. For an edge $e \in E$, $G-e$ will denote the graph obtained by deletion of the edge $e$. For $x,y \in V$, when $e=(x,y)=xy$, we may use $G-\{xy\}$ instead of $G-e$. Let $c:E\rightarrow \{1,2,\ldots,k\}$ be an \emph{acyclic edge coloring} of $G$. For an edge $e\in E$, $c(e)$ will denote the color given to $e$ with respect to the coloring $c$. For $x,y \in V$, when $e=(x,y)=xy$ we may use $c(x,y)$ instead of $c(e)$. For $S \subseteq V$, we denote the induced subgraph on $S$ by $G[S]$.

Many of the definitions, facts and lemmas that we develop in this section are already present in our earlier papers \cite{MBSC2}, \cite{MBSC4}. We include them here for the sake of completeness. The proofs of the lemmas will be omitted whenever it is available in \cite{MBSC2}, \cite{MBSC4}.

~~~~~~

\noindent \textbf{Partial Coloring:} Let H be a subgraph of $G$. Then an edge coloring $c'$ of $H$ is also a partial coloring of $G$. Note that $H$ can be $G$ itself. Thus a coloring $c$ of $G$ itself can be considered a partial coloring. A partial coloring $c$ of $G$ is said to be a proper partial coloring if $c$ is proper. A proper partial coloring $c$ is called acyclic if there are no bichromatic cycles in the graph. Sometimes we also use the word valid coloring instead of acyclic coloring. Note that with respect to a partial coloring $c$, $c(e)$ may not be defined for an edge $e$. So, whenever we use $c(e)$, we are considering an edge $e$ for which $c(e)$ is defined, though we may not always explicitly mention it.

Let $c$ be a partial coloring of $G$. We denote the set of colors in the partial coloring $c$ by $C = \{1,2,\ldots,k\}$. For any vertex $u \in V(G)$, we define $F_u(c) =\{c(u,z) \vert z \in N_{G}(u)\}$. For an edge $ab \in E$, we define $S_{ab}(c) = F_b - \{c(a,b)\}$. Note that $S_{ab}(c)$ need not be the same as $S_{ba}(c)$. We will abbreviate the notation to $F_u$ and $S_{ab}$ when the coloring $c$ is understood from the context.

To prove the main result, we plan to use contradiction. Let $G$ be the minimum counter example with respect to the number of edges for the statement in the theorems that we plan to prove.  Let $G =(V,E)$ be a graph on $m$ edges where $m \ge 1$. We will remove an edge $e=(x,y)$ from $G$ and get a graph $G'=(V,E')$. By the minimality of $G$, the graph $G'$ will have an acyclic edge coloring $c:E'\rightarrow \{1,2,\ldots,t\}$, where $t$ is the claimed upper bound for $a'(G)$. Our intention will be to extend the coloring $c$ of $G'$ to $G$ by assigning an appropriate color for the edge $e$ thereby contradicting the assumption that $G$ is a minimum counter example.

~~~~~

The following definitions arise out of our attempt to understand what may prevent us from extending a partial coloring of $G-e$ to $G$.

\noindent \textbf{Maximal bichromatic Path:} An ($\alpha$,$\beta$)-maximal bichromatic path with respect to a partial coloring $c$ of $G$ is a maximal path consisting of edges that are colored using the colors $\alpha$ and $\beta$ alternatingly. An ($\alpha$,$\beta$,$a$,$b$)-maximal bichromatic path is an ($\alpha$,$\beta$)-maximal bichromatic path which starts at the vertex $a$ with an edge colored $\alpha$ and ends at $b$. We emphasize that the edge of the ($\alpha$,$\beta$,$a$,$b$)-maximal bichromatic path incident on vertex $a$ is colored $\alpha$ and the edge incident on vertex $b$ can be colored either $\alpha$ or $\beta$. Thus the notations ($\alpha$,$\beta$,$a$,$b$) and ($\alpha$,$\beta$,$b$,$a$) have different meanings. Also note that any maximal bichromatic path will have at least two edges. The following fact is obvious from the definition of proper edge coloring:

\begin{fact}
\label{fact:fact1}
Given a pair of colors $\alpha$ and $\beta$ of a proper coloring $c$ of $G$, there can be at most one maximal ($\alpha$,$\beta$)-bichromatic path containing a particular vertex $v$, with respect to $c$.
\end{fact}

A color $\alpha \neq c(e)$ is a \emph{candidate} for an edge \emph{e} in $G$ with respect to a partial coloring $c$ of $G$ if none of the adjacent edges of \emph{e} are colored $\alpha$. A candidate color $\alpha$ is \emph{valid} for an edge \emph{e} if assigning the color $\alpha$ to \emph{e} does not result in any bichromatic cycle in $G$.

Let $e=(a,b)$ be an edge in $G$. Note that any color $\beta \notin F_a \cup F_b$ is a candidate color for the edge $ab$ in $G$ with respect to the partial coloring $c$ of $G$. A sufficient condition for a candidate color being valid is captured in the Lemma below (See Appendix for proof):

\begin{lem}
\label{lem:lem1}
\cite{MBSC2}
A candidate color for an edge $e=ab$, is valid if $(F_a \cap F_b) - \{c(a,b)\} = (S_{ab} \cap S_{ba}) = \emptyset$.
\end{lem}

Now even if $S_{ab} \cap S_{ba} \neq \emptyset$, a candidate color $\beta$ may be valid. But if $\beta$ is not valid, then what may be the reason? It is clear that color $\beta$ is not $valid$ if and only if there exists $\alpha \neq \beta$ such that a ($\alpha$,$\beta$)-bichromatic cycle gets formed if we assign color $\beta$ to the edge $e$. In other words, if and only if, with respect to coloring $c$ of $G$ there existed a ($\alpha$,$\beta$,$a$,$b$) maximal bichromatic path with $\alpha$ being the color given to the first and last edge of this path. Such paths play an important role in our proofs. We call them $critical\ paths$. It is formally defined below: \newline

\noindent\textbf{Critical Path:} Let $ab \in E$ and $c$ be a partial coloring of $G$. Then a $(\alpha,\beta,$a$,$b$)$ maximal bichromatic path which starts out from the vertex $a$ via an edge colored  $\alpha$ and ends at the vertex $b$ via an edge colored $\alpha$ is called an $(\alpha,\beta,ab)$ critical path. Note that any critical path will be of odd length. Moreover the smallest length possible is three.

~~~~~~~~~
%

An obvious strategy to extend a valid partial coloring $c$ of $G$ would be to try to assign one of the candidate colors to an uncolored edge $e$. The condition that a candidate color being not valid for the edge $e$ is captured in the following fact.

\begin{fact}
\label{fact:fact2}
Let $c$ be a partial coloring of $G$. A candidate color $\beta$ is not $valid$ for the edge $e=(a,b)$ if and only if $\exists \alpha \in S_{ab} \cap S_{ba}$ such that there is a $(\alpha,\beta,ab)$  critical path in $G$ with respect to the coloring $c$.
\end{fact}

\noindent \textbf{Actively Present:} Let $c$ be a partial coloring of $G$. Let $a \in N_{G}(x)$ and let $c(x,a)=\alpha$. Let $\beta \in S_{xa}$. Color $\beta$ is said to be \emph{actively present} in a set $S_{xa}$ with respect to the edge $xy$, if there exists a $(\alpha,\beta,xy)$  critical path. When the edge $xy$ is understood in the context, we just say that $\beta$ is actively present in $S_{xa}$.

~~~~~

%

\noindent \textbf{Color Exchange:} Let $c$ be a partial coloring of $G$. Let $u,i,j \in V(G)$ and $ui,uj \in E(G)$. We define $Color\ Exchange$ with respect to the edge $ui$ and $uj$, as the modification of the current partial coloring $c$ by exchanging the colors of the edges $ui$ and $uj$ to get a partial coloring $c'$, i.e., $c'(u,i)=c(u,j)$, $c'(u,j)=c(u,i)$ and $c'(e)=c(e)$ for all other edges $e$ in $G$. The color exchange with respect to the edges $ui$ and $uj$ is said to be proper if the coloring obtained after the exchange is proper. The color exchange with respect to the edges $ui$ and $uj$ is $valid$ if and only if the coloring obtained after the exchange is acyclic. The following fact is obvious:

\begin{fact}
\label{fact:fact3}
Let $c'$ be the partial coloring obtained from a valid partial coloring $c$ by the color exchange with respect to the edges $ui$ and $uj$. Then the partial coloring $c'$ will be proper if and only if $c(u,i) \notin S_{uj}$ and $c(u,j) \notin S_{ui}$.
\end{fact}

~~~~~~

\noindent The color exchange is useful in breaking some critical paths as is clear from the following lemma (See Appendix for proof):

\begin{lem}
\label{lem:lem2}
\cite{MBSC2}, \cite{MBSC4}
Let $u,i,j,a,b \in V(G)$, $ui,uj,ab \in E$. Also let $\{\lambda,\xi\} \in C$ such that $\{\lambda,\xi\} \cap \{c(u,i),c(u,j)\} \neq \emptyset$ and $\{i,j\} \cap \{a,b\} = \emptyset$. Suppose there exists an ($\lambda$,$\xi$,$ab$)-critical path that contains vertex $u$, with respect to a valid partial coloring $c$ of $G$. Let $c'$ be the partial coloring obtained from $c$ by the color exchange with respect to the edges $ui$ and $uj$. If $c'$ is proper, then there will not be any ($\lambda$,$\xi$,$ab$)-critical path in $G$ with respect to the partial coloring $c'$.
\end{lem}

\noindent The following is the main result of \cite{MBSC2}. We will need this result for proving our theorems.

\begin{lem}
\label{lem:lem3}
\cite{MBSC2}
Let $G$ be a connected graph on $n$ vertices, $m \le 2n-1$ edges and maximum degree $\Delta \le 4$, then $a'(G)\le 6$.
\end{lem}

\section{Proof of Theorem 1}

\begin{proof}
A well-known strategy that is used in proving coloring theorems in the context of sparse graphs is to make use of induction combined with the fact that there are some \emph{unavoidable} configurations in any such graphs. Typically the existence of these \emph{unavoidable} configurations are proved using the so called \emph{charging and discharging argument} (See \cite{Salavatipour}, for a comprehensive exposition). $Lemma$ \ref{lem:lem4} will establish that one of the five configurations $B1, \ldots,B5$ is unavoidable in any graph $G$ that satisfies $Property\ A$. Loosely speaking, for the purpose of this paper, a \emph{configuration} is a subset $Q$ of $V$, where one special vertex $v \in Q$ is called the $pivot$ of the configuration and $Q = \{v\} \cup N(v)$. Besides $v$, one more vertex in $Q$ will be given a special status: This vertex, called the \emph{co-pivot} of the configuration, is selected such that it is a vertex of smallest degree in $N(v)$ and will be denoted by $u$. Moreover the vertices of $N(v)$ will be partitioned into two sets namely $N'(v)$ and $N''(v)$. The members of $N'(v)$ and $N''(v)$ are explicitly defined for each configuration.

\begin{lem}
\label{lem:lem4}
Let G be a simple graph such that $\vert E(G)\vert \le 2\vert V(G)\vert-1$ with minimum degree $\delta \ge 2$. Then there exists a vertex $v$ in $G$ with $k=deg(v)$ neighbours such that at least one of the following is true:
\begin{enumerate}
\item[(B1)] $k=2$,
\item[(B2)] $k=3$ with $N(v)=\{u,v_1,a\}$ such that $deg(u),deg(v_1) \le 4$. $N'(v)=\{u,v_1\}$ and $N''(v)=\{a\}$,
\item[(B3)] $k=5$ with $N(v)=\{u,v_1,v_2,a,b\}$ such that $deg(u),deg(v_1),deg(v_2) \le 3$. $N'(v)=\{u,v_1,v_2\}$ and $N''(v)=\{a,b\}$,
\item[(B4)] $k=6$ with $N(v)=\{u,v_1,v_2,v_3,v_4,a\}$ such that $deg(u),deg(v_1),deg(v_2),deg(v_3),deg(v_4) \le 3$. $N'(v)=\{u,v_1,v_2,v_3,v_4\}$ and $N''(v)=\{a\}$,
\item[(B5)] $k \ge 7$ with $N(v) = \{u,v_1,v_2,\dots,v_{k-1}\}$ such that $deg(u),deg(v_1),deg(v_2), \ldots ,deg(v_{k-1}) \le 3$. $N'(v)=\{u,v_1,v_2,\ldots, v_{k-1}\}$.
\end{enumerate}
\end{lem}
\begin{proof}
We use the discharging method to prove the lemma. Let $G=(V,E)$, $\delta \ge 2$, $\vert V \vert =n$ and $\vert E \vert =m \le 2n-1$. We define  a mapping $\phi:V \longmapsto \mathbbm{R}$ using the rule $\phi(v)= deg(v)-4$ for each $v \in V$. The value $\phi(v)$ is called the charge on the vertex $v$. Since $m \le 2n-1$, it is easy to see that $\sum_{v \in V}\phi(v) \le -2$. Now we redistribute the charges on the vertices using the following rule. (This procedure is usually known as \emph{discharging}: Note that the total charge has to remain same after the discharging.)
\begin{itemize}
\item If vertex $v$ has degree at least 5, then it gives a charge of $\frac{1}{2}$ to each of its 3-degree neighbours.
\end{itemize}

After $discharging$, each vertex $v$ has a new charge $\phi'(v)$. Now since the total charge is conserved, we have $\sum_{v \in V}\phi(v) = \sum_{v \in V}\phi'(v) \le -2$. Now suppose the graph $G$ has  none of the configurations $B1,\dots,B5$. Then we will show that for each vertex $v$ of $G$, $\phi'(v) \ge 0$ and therefore $\sum_{v \in V}\phi'(v) \ge 0$, a contradiction. Since $G$ does not have configuration $B1$, we have $\delta \ge 3$. Now we calculate the charge on each vertex $v$ of $G$ as follows:

\begin{itemize}
\item If $deg(v)=3$: Since $G$ does not have configuration $B2$, at least two of the neighbours have degree at least 5. Thus $v$ receives a charge of $\frac{1}{2}$ each from at least two of its neighbours. Thus $\phi'(v) \ge deg(v)-4+2 \cdot \frac{1}{2} = 0$.

\item If $deg(v)=4$: A four degree vertex does not give or receive any charge. Thus $\phi'(v)=\phi(v)=deg(v)-4=0$.

\item If $deg(v)=5$: Since $G$ does not have configuration B3, at most two of the neighbours have degree 3. Thus $v$ gives a charge of $\frac{1}{2}$ each to at most two of its neighbours. Thus $\phi'(v) \ge deg(v)-4-2 \cdot \frac{1}{2} = 0$.

\item If $deg(v)=6$: Since $G$ does not have configuration B4, at most four of the neighbours have degree 3. Thus $v$ gives a charge of $\frac{1}{2}$ each to at most four of its neighbours. Thus $\phi'(v) \ge deg(v)-4-4 \cdot \frac{1}{2} = 0$.

\item If $deg(v)\ge 7$: Since $G$ does not have configuration B5, at most $deg(v)-1$ of the neighbours have degree 3. Thus $v$ gives a charge of $\frac{1}{2}$ each to at most $deg(v)-1$ of its neighbours. Thus $\phi'(v) \ge deg(v)-4-(deg(v)-1) \cdot \frac{1}{2} = \frac{1}{2}(deg(v)-7) \ge 0$.
\end{itemize}

Thus we have established that $\phi'(v) \ge 0$, $\forall v \in V$ and therefore $\sum_{v \in V}\phi'(v) \ge 0$, a contradiction.
\end{proof}

~~~~~

We prove the theorem by way of contradiction. Let $G$ be a minimum counter example (with respect to the number of edges) for the theorem statement among the graphs satisfying $Property\ A$. Clearly $G$ is $2$-$connected$ since if there are cut vertices in $G$, the acyclic edge coloring of the blocks $G_1,G_2,\ldots, G_k$ of $G$ can easily be extended to $G$ (Note that each block satisfies the $Property\ A$ since they are subgraphs of $G$). Thus we have, $\delta(G)\ge 2$. Also from $Lemma$ \ref{lem:lem3}, we know that $a'(G) \le \Delta +3$, when $\Delta \le 4$. Therefore we can assume that $\Delta \ge 5$. Thus we have,

\begin{asm}
\label{asm:asm1}
For the minimum counter example $G$, $\delta(G) \ge 2$ and $\Delta(G) \ge 5$.
\end{asm}

\noindent By $Lemma$ \ref{lem:lem4}, graph $G$ has a vertex $v$, such that it is the pivot of one of the configurations $B1,\dots,B5$. We present the proof in two parts based on the configuration that $v$ belongs to. The first part deals with the case when $G$ has a vertex $v$ that belongs to configuration $B2$, $B3$, $B4$ or $B5$ and the second part deals with the case when $G$ does not have a vertex $v$ that belongs to configuration $B2$, $B3$, $B4$ or $B5$.

\subsection{There exists a vertex $v$ that belongs to configuration $B2$, $B3$, $B4$ or $B5$}

Let $v$ be a vertex such that it is the pivot of one of the configurations $B2,\ldots,B5$ and let $u$ be the co-pivot. Since $G$ is a minimum counter example, the graph $G-\{vu\}$ is acyclically edge colorable using $\Delta +3$ colors. Let $c'$ be a valid coloring of $G-\{vu\}$ and hence a partial coloring of $G$. We now try to extend $c'$ to a valid coloring of $G$. With respect to the partial coloring $c'$ let $F'_{v}(c')=\{c'(v,x) \vert x \in N'(v)\}$ and $F''_{v}(c')=\{c'(v,x) \vert x \in N''(v)\}$ i.e., $F''_{v}=F_v-F'_v$.

\begin{clm}
\label{clm:clm1}
With respect to any valid coloring $c'$ of $G-\{uv\}$, $\vert F_{u} \cap F_{v} \vert \ge 1$
\end{clm}
\begin{proof}
Suppose not. Then $S_{vu} \cap S_{uv} = \emptyset$ and by $Lemma$ \ref{lem:lem1}, all the candidate colors are valid for the edge $vu$. It is easy to verify that irrespective of which configuration $v$ belongs to, $\vert F_{u} \cup F_{v} \vert \le \Delta-1 +2 = \Delta +1$. Therefore there are at least two candidate colors for the edge $vu$ which are also valid, a contradiction to the assumption that $G$ is a counter example.
\end{proof}

\begin{clm}
\label{clm:clm2}
$\forall x \in N(v)$, we have $deg(x) \ge 3$.
\end{clm}
\begin{proof}
Suppose not. Then by $Assumption$ \ref{asm:asm1}, it is clear that the degree of the co-pivot, $deg(u) = 2$. Let $N(u)=\{v,v'\}$. It is easy to verify from the description of configurations $B2-B5$ and the fact that $deg(u)=2$ that there can be at most two vertices in $N(v)$ whose degrees are greater than 3. By Claim \ref{clm:clm1}, we know that $c'(u,v') \in F_v$. Let $D_v =D_{v}(c')=\{c'(v,x) \vert deg_{G}(x) \le 3\}$. Clearly have $\vert D_v \vert \le 2$.

If $c'(u,v') \in F_v-D_v$, then let $c=c'$. Else if $c'(u,v') \in D_v$, then recolor edge $uv'$ using a color from $C-(S_{uv'} \cup D_v)$ to get a coloring $c$ (Note that $\vert C-(S_{uv'} \cup D_v) \vert \ge \Delta+3-(\Delta-1+2) = 2$ and since $u'$ is a pendant vertex in $G-\{uu'\}$ the recoloring is valid). Now if $c(u,v') \notin F_v$, then it a contradiction to Claim \ref{clm:clm1}. Thus $c(u,v') \in F_v-D_v$.

With respect to coloring $c$, let $c(u,v')=c(v,v_1)$. Now there are at least four candidate colors for the edge $uv$ since $\vert F_u \cup F_{v} \vert \le \Delta -1$. If none of them are valid then they all have to be actively present in $S_{vv_1}$, implying that $\vert S_{vv_1} \vert \ge 4$, a contradiction since $\vert S_{vv_1} \vert \le 3$. Thus there exists a color valid for the edge $uv$, a contradiction to the assumption that $G$ is a counter example.
\end{proof}

\begin{clm}
\label{clm:clm3}
$deg(v) > 3$. Therefore $v$ does not belong to $Configuration\ B2$.
\end{clm}
\begin{proof}
Suppose $v$ belongs to $Configuration\ B2$. Let $N(v)=\{u,v_1,a\}$ such that $deg(u) \le 4$ and $deg(v_1) \le 4$. We also know from $Claim$ \ref{clm:clm2} that $deg(u) \ge 3$. Let $N(u)= \{x,y,v\}$, if $deg(u)=3$ and let $N(u)=\{x,y,z,v\}$, if $deg(u)=4$. Now the following cases occur:

\begin{itemize}
\item $\vert F_{u} \cap F_{v} \vert= 2$. \newline
Let $F_{u} \cap F_{v} = \{1,2\}$. Also let  $c(u,x)=c(v,a)=1$ and $c(u,y)=c(v,v_1)=2$. Since $\vert F_v \cup F_{u} \vert \le 3$, there are at least $\Delta$ candidate colors for the edge $vu$. If none of them are valid then all those colors are actively present either in $S_{vv_1}$ or $S_{va}$. Recalling that $\vert S_{va} \vert \le \Delta-1$ we can infer that there is at least one color $\alpha \in C-(F_v \cup F_{u})$ that does not belong to $S_{va}$. Note that $\vert S_{vv_1} \cup F_v \cup F_{u} \vert \le 6$ since $\vert S_{vv_1} \vert \le 3$ and $\vert F_v \cup F_{u} \vert \le 3$. Since $\Delta \ge 5$, we have $C- (S_{vv_1} \cup F_v \cup F_{u}) \neq \emptyset$. Recolor the edge $vv_1$ with the a color $\beta$ from $C- (S_{vv_1} \cup F_v \cup F_{u})$ to get a coloring $c$. The coloring $c$ is valid because if a bichromatic cycle gets created due to recoloring then it has to be a $(\beta,1)$ bichromatic cycle since $c(v,a)=1$, implying that there existed a $(1,\beta,vv_1)$ critical path with respect to coloring $c'$. Recall that color $\beta$ was not valid for the edge $vu$. Since $\beta \notin S_{vv_1}$, it implies that color $\beta$ was actively present in $S_{va}$. This implies that there existed a $(1,\beta,vu)$ critical path with respect to coloring $c'$. Therefore by $Fact$ \ref{fact:fact1}, there cannot exists a $(1,\beta,vv_1)$ critical path with respect to $c'$, a contradiction. Thus the coloring $c$ is valid. Now in $c$ we have $F_v \cap F_{u} =\{1\}$ and $\alpha \notin S_{va}$. Thus color $\alpha$ is valid for the edge $vu$, a contradiction to the assumption that $G$ is a counter example.

\item $\vert F_{u} \cap F_{v} \vert= 1$. \newline
Let $F_{u} \cap F_{v} = \{1\}$. Now if $c'(v,v_1) \in F_{u} \cap F_{v}$, then let $c''=c'$. Otherwise let $c(u,x)=c(v,a)=1$ and $c'(v,v_1)=4$. If $deg(u) \le 3$, then $\vert F_{v} \cup F_{u} \vert = 3$. Now there are at least $\Delta$ candidate colors for the edge $vu$. If none of them are valid then all the candidate colors are actively present in $S_{va}$, a contradiction since $\vert S_{va} \vert \le \Delta -1$. Thus there exists a valid color for the edge $vu$. Thus $deg(u) = 4$ and $\vert F_{v} \cup F_{u} \vert = 4$. Let $c(u,y)=2$ and $c(u,z)=3$. There are at least $\Delta-1$ candidate colors for the edge $vu$. If none of them are valid then all the candidate colors are actively present in $S_{va}$ and $S_{ux}$, implying that $S_{va}=S_{ux}= C- \{1,2,3,4\}$. Now recolor edge $ux$ using color $4$ to get a coloring $c''$. It is valid by $Lemma$ \ref{lem:lem1} since $S_{ux} \cap S_{xu}= \emptyset$ (Note that $S_{xu}(c')=\{2,3\}$).

In both cases we have $\{c''(v,v_1)\} = F_{u} \cap F_{v}$. If none of the colors are valid for the edge $vu$, then all the candidate colors are actively present in $S_{vv'}$, implying that $S_{vv_1}= C- \{1,2,3,4\}$. Since $\Delta \ge 5$, we have $\vert C-\{1,2,3,4\} \vert \ge 8-4= 4$. But $\vert S_{vv_1} \vert \le 3$, a contradiction. Thus there exists a color valid for the edge $vu$, a contradiction to the assumption that $G$ is a counter example.
\end{itemize}

\end{proof}


~~~~~
In view of $Claim$ \ref{clm:clm3} we have $deg(v) > 3$. Therefore $v$ belongs to configurations $B3$, $B4$ or $B5$. Now in view of $Claim$ \ref{clm:clm2}, we have the following observation:

\begin{obs}
\label{obs:obs1}
$deg(u) = 3$. Let $N(u) =\{v,w,z\}$.
\end{obs}

In view of $Claim$ \ref{clm:clm1}, we have the following two cases:

\subsubsection{case 1: $\vert F_{v} \cap F_{u} \vert= 2$}

Note that in this case $F_u \subseteq F_v$. Let $F_u = F_u \cap F_v=\{1,2\}$. Let $c'(u,z)=1$ and $c'(u,w)=2$.

~~~~~

\begin{clm}
\label{clm:clm4}
$F_{u} \nsubseteq  F'_v$. Therefore $F''_v \cap F_{u} \neq \emptyset$.
\end{clm}
\begin{proof}
Suppose not. Then let $c'(v,v_1)=c'(u,z)=1$ and $c'(v,v_2)=c'(u,w)=2$ (See the statement of $Lemma$ \ref{lem:lem4} for the naming convention of the neighbours of $v$). Since $\vert F_{u} \cup F_{v} \vert \le \Delta -1$, there are at least four candidate colors for the edge $vu$. If none of the candidate colors are valid for the edge $vu$, then we should have $S_{vv_1} \subset C-(F_{u} \cup F_v)$ and $S_{vv_2} \subset C-(F_{u} \cup F_v)$ since $\vert S_{vv_1} \vert = 2$ and $\vert S_{vv_2} \vert = 2$. Also $S_{vv_1} \cap S_{vv_2} = \emptyset$. Note that $C- (S_{vv_1} \cup F_v \cup F_{u}) \neq \emptyset$ since $\vert F_{u} \cup F_{v} \vert \le \Delta -1$ and $\vert S_{vv_1} \vert = 2$. Now assign a color from  $C- (S_{vv_1}\cup F_{u} \cup F_v)$ to the edge $vv_1$ to get a coloring $c$. Recall that $S_{vv_1} \subset C-(F_{u} \cup F_{v})$ and therefore $S_{vv_1} \cap S_{v_1v} = \emptyset$. Thus by $Lemma$ \ref{lem:lem1}, the coloring $c$ is valid. With respect to the coloring $c$, $F_{u} \cap F_{v} = \{2\}$ and therefore if a candidate color is not valid for the edge $vu$, it has to be actively present in $S_{vv_2}$. Let $\alpha \in S_{vv_1}$. Clearly $\alpha \in C-(F_{u} \cup F_{v})$ is a candidate color for the edge $vu$. Now since $\alpha \notin S_{vv_2}$ (recall that $S_{vv_1} \cap S_{vv_2} = \emptyset$), color $\alpha$ is valid for the edge $vu$, a contradiction to the assumption that $G$ is a counter example.
\end{proof}

~~~~~~

In view of $Claim$ \ref{clm:clm4}, $F''_v \cap F_{u} \neq \emptyset$ and therefore $F''_v \neq \emptyset$. It follows that vertex $v$ does not belong to configuration $B5$. Recalling $Claim$ \ref{clm:clm3}, we infer that the vertex $v$  belongs to either configuration $B3$ or $B4$. We take care of these two configurations separately below:

~~~~

\noindent \textbf{subcase 1.1: $v$ belongs to configuration $B3$.} \newline
Since $deg(v) = 5$, we have $\vert F_v \vert = 4$. Let $F_v =\{1,2,3,4\}$. Recall that by Claim \ref{clm:clm4}, we have $F''_v \cap F_{u} \neq \emptyset$. Without loss of generality let $c'(u,z)=c'(v,a)=1$ and  $c'(u,w)=2$. Now there are $\Delta-1$ candidate colors for the edge $vu$. If none of them are valid then all these candidate colors are actively present in at least one of $S_{uz}$ and $S_{uw}$. Let $Y = C-\{1,2,3,4\}$. We make the following claim:

\begin{clm}
\label{clm:clm5}
With respect to any valid coloring $c'$ of $G-\{uv\}$, $Y =S_{uz}$ and $Y =S_{uw}$.
\end{clm}
\begin{proof}
We use contradiction to prove the claim. Firstly we make the following subclaim:

~~~~

\noindent \textbf{subclaim {\ref{clm:clm5}}.1:} \emph{With respect to any valid coloring $c'$ of $G-\{uv\}$, one of $S_{uz}$ or $S_{uw}$ is $Y$.}
\begin{proof}
Suppose not. Then $Y \neq S_{uz}$ and $Y \neq S_{uw}$. Note that $\vert Y \vert = \Delta -1$ while $\vert S_{uz} \vert \le \Delta -1$ and $\vert S_{uw} \vert \le \Delta -1$. Therefore there exist colors $\alpha, \beta \in Y$ such that $\alpha \notin S_{uz}$ and $\beta \notin S_{uw}$. Note that $\alpha \neq \beta$ since otherwise color $\alpha=\beta$ will be valid for the edge $vu$ as there cannot exist a $(1,\alpha,vu)$ or $(2,\alpha,vu)$ critical path with respect to $c'$. It follows that $\alpha$ is actively present in $S_{uw}$ and $\beta$ is actively present in $S_{uz}$. Hence there exist $(2,\alpha,vu)$ and $(1,\beta,vu)$ critical paths. Now recolor edge $uz$ using color $\alpha$ to get a coloring $c''$. The recoloring is valid since if there is a bichromatic cycle then it has to be a $(\alpha,2)$ bichromatic cycle, implying that there existed a $(2,\alpha,uz)$ critical path in $c'$, a contradiction in view of Fact $\ref{fact:fact1}$ as there already existed a $(2,\alpha,vu)$ critical path. With respect to coloring $c''$, $F_{v} \cap F_{u}=\{2\}$ and therefore if a candidate color is not valid for the edge $vu$, it has to be actively present in $S_{uw}$. Now color $\beta \notin S_{uw}$ and hence color $\beta$ is valid for the edge $vu$, a contradiction to the assumption that $G$ is a counter example.
\end{proof}

With respect to any valid coloring $c'$ of $G-\{uv\}$, in view of $subclaim$ \ref{clm:clm5}.1, let $u' \in \{w,z\}$ be such that $S_{uu'}=Y$. Let $\{u''\} =\{w,z\}-\{u'\}$. Now for contradiction assume that $S_{uu''} \neq Y$. Then clearly there exists a color $\alpha \in Y$ such that $\alpha \notin S_{uu''}$.

~~~~~

\noindent \textbf{subclaim \ref{clm:clm5}.2:} \emph{With respect to any valid coloring $c'$ of $G-\{uv\}$, if exactly one of $S_{uw}$ and $S_{uz}$ is $Y$, say $S_{uu'}=Y$, then all the colors of $Y$ are actively present in $S_{uu'}$ and $c'(u,u') \in F''_v$.}
\begin{proof}
Recolor the edge $uu''$ with the color $\alpha$ to get a coloring $c''$. Since $\alpha \notin S_{uu''}$ and $\alpha$ is not valid for the edge $vu$, color $\alpha$ is actively present in $S_{uu'}$ i.e., with respect to coloring $c'$, there exists a $(\gamma,\alpha,vu)$ critical path, where $\gamma =c'(u,u')$. Thus by $Fact$ \ref{fact:afct1}, there cannot exist a $(\gamma,\alpha,uu'')$ critical path and hence the coloring $c''$ is valid for the edge $uu''$. With respect to coloring $c''$, $F_{v} \cap F_{u}=\{2\}$. Now all the $\Delta-2$ colors from $Y-\{\alpha\}$ are candidates for the edge $vu$. If any one of them is valid we are done. Thus none of them are valid and hence they all have to be actively present in $S_{uu'}$. Recalling that the color  $\alpha$ was actively present in $S_{uu'}$ we infer that all the colors of $Y$ are in fact  actively present in $S_{uu'}$.

Now these colors will also be actively present in $S_{vv'}$, where $v' \in N(v)$ is such that $c'(v,v')=c'(u,u')$. This implies that $\vert S_{vv'} \vert = \vert Y \vert = \Delta -1$. Therefore $v'$ cannot be $v_1$ or $v_2$ since $\vert S_{vv_1} \vert = 2$ and $\vert S_{vv_2} \vert = 2$ while $\Delta-1 \ge 4$. Thus $v' \in N''(v)$ implying that $c'(u,u') \in F''_v$.
\end{proof}

Recalling that for configuration $B3$, $\vert F''_v \vert = 2$ and since $1 \in F''_v$, at least one of $3,4$ belongs to $F'_v$. Without loss of generality let $3 \in F'_v$. Now recolor edge $uu'$ using color $3$ to get a coloring $d$ from $c'$. The coloring $d$ is valid by $Lemma$ \ref{lem:lem1} since $\{d(u,u'')\} \cap S_{uu'} = \{2\} \cap Y = \emptyset$. With respect to the coloring $d$ we have $S_{uu'}=Y$ and $S_{uu''} \neq Y$. Thus by $subclaim$ \ref{clm:clm5}.2, $d(u,u') \in F''_v$, a contradiction since $d(u,u')=3 \notin F''_v$. Thus we have $Y =S_{uz}$ and $Y =S_{uw}$.
\end{proof}

Since $Y =S_{uz}$ and $Y =S_{uw}$, we can recolor edge $uz$ and $uw$ using color from $F'_v$ (Recall that with respect to configuration $B3$, $\vert F'_v \vert = 2$) to get a new valid coloring $c$. The coloring $c$ is valid by $Lemma$ \ref{lem:lem1} since $F'_v \cap S_{uz} = F'_v \cap Y = \emptyset$ and $F'_v \cap S_{uw} = F'_v \cap Y = \emptyset$. This reduces the situation to $F_{u} \subseteq  F'_v$, a contradiction to $Claim$ \ref{clm:clm4}.

~~~~~~~~

\noindent \textbf{subcase 1.2: $v$ belongs to configuration $B4$.} \newline
~~~~~~

We have $deg(v)=6$ and  $F''_v=\{c'(v,a)\}$. Therefore in view of Claim \ref{clm:clm4}, $c'(v,a)$ has to belong to $F_{u}$. Let $F_v =\{1,2,3,4,5\}$. Without loss of generality let $c'(u,w)=c'(v,v_1)=2$ and $c'(u,z)=c'(v,a)=1$. Now there are $\Delta-2$ candidate colors for the edge $vu$. If none of them are valid then all these candidate colors are actively present in at least one of $S_{uz}$ and $S_{uw}$. Let $X = C-\{1,2,3,4,5\}$.

\begin{clm}
\label{clm:clm6}
$X \subseteq S_{uz}$.
\end{clm}
\begin{proof}
Suppose not. Then let $\alpha$ be a color such that $\alpha \in X-S_{uz}$. This implies that $\alpha$ is actively present in $S_{uw}$. Hence there exists a $(2,\alpha,vu)$ critical path since $c'(u,w)=2$. Now recolor edge $uz$ using color $\alpha$ to get a coloring $c''$. The recoloring is valid since if there is a bichromatic cycle then it has to be a $(\alpha,2)$ bichromatic cycle, implying that there existed a $(2,\alpha,uz)$ critical path in $c'$, a contradiction in view of Fact $\ref{fact:fact1}$ as there already existed a $(2,\alpha,vu)$ critical path. Now with respect to coloring $c''$, $F_v \cap F_u = \{2\}$ and therefore if none of the colors in $X-\{\alpha\}$ is valid for the edge $vu$, they all should be actively present in $S_{uw}$. Recalling that color $\alpha$ was actively present in $S_{uw}$ we have all the colors of $X$ actively present in $S_{uw}$ and hence in $S_{vv_1}$ implying that $ \vert S_{vv_1} \vert \ge \vert X \vert = \Delta -2 \ge 3$, a contradiction since $\vert S_{vv_1} \vert = 2$. Thus there exists a color valid for the edge $vu$, a contradiction to the assumption that $G$ is a counter example.
\end{proof}

\begin{clm}
\label{clm:clm7}
$X \subseteq S_{uw}$.
\end{clm}
\begin{proof}
Suppose not. Then let $X \nsubseteq S_{uw}$ and let $\alpha$ be a color such that $\alpha \in X-S_{uw}$. Recolor the edge $uw$ using the color $\alpha$. It is easy to see (by a similar argument used in the proof of Claim \ref{clm:clm6}) that $c''$ is valid and all the colors of $X$ are actively present in $S_{uz}$ and hence in $S_{va}$.

Since $\vert X \vert = \Delta -2$ and $\vert S_{va} \vert \le \Delta -1$, we have $\vert S_{va} - X \vert \le 1$. If $S_{va} \neq X$, then the singleton set $S_{va}-X$ has to be a subset of $\{2,3,4,5\}$ since $1 \notin S_{va}$. Without loss of generality let $S_{va}-X=\{2\}$ (Reader may note that $\{2,3,4,5\}=F'_v$ and these four colors play symmetric roles in $c''$ and therefore we need to argue with respect to only one of them). Recall that $c''(v,v_1)=c'(v,v_1)=2$ and $\vert S_{vv_1}\vert =2$. Of the colors $3$, $4$ and $5$ let $3 \notin S_{vv_1}$. Also let $c''(v,v_2)=3$. Now delete the color on the edge $vv_2$ and recolor the edge $va$ using  color $3$ to get a coloring $d$. We claim that the coloring $d$ is valid: If $S_{va}=X$, then clearly it is valid by $Lemma$ \ref{lem:lem1} since $S_{va} \cap S_{av} = \emptyset$. Otherwise we have $S_{va}-X=\{2\}$ and if there is a bichromatic cycle with respect to the coloring $d$, it has to be a $(2,3)$ bichromatic cycle. Since $d(v,v_1)=2$, it means that $3 \in S_{vv_1}$, a contradiction to our assumption. Thus the coloring $d$ is valid.

Now with respect to coloring $d$, we have $d(u,z)=1$, $d(u,w)=\alpha$, $d(v,a)=3$, $d(v,v_1)=2$, $d(v,v_3)=4$ and $d(v,v_4)=5$. Edges $vu$ and $vv_2$ are uncolored. Now let $X' = C-\{2,3,4,5\}$. Note that $\vert X' \vert \ge 5$ since $\Delta \ge 6$. We show below that there exists a color in $X'$ that is valid for the edge $vv_2$:

\begin{itemize}
\item $S_{vv_2} \subset X'$. Now any color in $X'-S_{vv_2}$ is valid for the edge $vv_2$ by $Lemma$ \ref{lem:lem1}.

\item $\vert S_{vv_2} \cap X' \vert =1$. In this case exactly one color, say $\theta \in \{2,4,5\}$ is present in $S_{vv_2}$ since $3 \notin S_{vv_2}$ (This is because $c'(v,v_2)=3$). Now there are at least four candidate colors for the edge $vv_2$ since $\vert F_v \cup F_u \vert \le 4+2-1 = 5$ and there are at least $\Delta+3 \ge deg(v)+3=6+3=9$ colors in $C$. If none of the candidate colors are valid then a $(\theta,\gamma)$ bichromatic cycle should form for each $\gamma \in X'-S_{vv_2}$. Since $\theta \in \{2,4,5\}$, we have $\theta = d(v,v_j)$ for $j=1$, $3$ or $4$. It means that each of the $(\theta,\gamma)$ bichromatic cycle should contain the edge $vv_j$ and thus $X'-S_{vv_2} \subseteq S_{vv_j}$. But $\vert X'-S_{vv_2} \vert \ge 5-2+1 \ge 4$ and  $\vert S_{vv_j} \vert= 2 $, a contradiction. Thus at least one color will be valid for the edge $vv_2$.

\item $S_{vv_2} \cap X' = \emptyset$. Now all the colors in $X'$ are candidates for the edge $vv_2$. If none of them are valid then all these candidate colors have to form bichromatic cycles with at least one of the colors in $S_{vv_2} \cap F_v$. Now since $c''(v,v_2)=3$, color $3 \notin S_{vv_2}(d)$ and therefore $3$ is not involved in any of these bichromatic cycles. Also since $\vert S_{vv_2} \vert = 2 $, exactly two of the colors from $\{2,4,5\}$ and hence exactly two of the edges from $\{vv_1, vv_3, vv_4\}$ are involved in these bichromatic cycles. But we know that $\vert S_{vv_1} \vert =\vert S_{vv_3} \vert = \vert S_{vv_4} \vert= 2 $. It follows that at most four bichromatic cycles can be formed. But $\vert X' \vert \ge 5$ and thus at least one color will be valid for the edge $vv_2$.
\end{itemize}

\noindent Let $\beta \in X'$ be a valid color for $vv_2$. Color the edge $vv_2$ using $\beta$ to get a new coloring $d'$. Now:

\begin{itemize}
\item If $\beta \in C-\{1,2,3,4,5,\alpha\}$, then $F_v \cap F_u = \emptyset$ with respect to $d'$, a contradiction to $Claim$ \ref{clm:clm1}.

\item If $\beta \in \{1,\alpha\}$, then there are at least three candidate colors for the edge $vu$ since $\Delta \ge 6$. Moreover we have $F_v \cap F_u =\{\beta\}$. If none of these three candidate colors are valid for the edge $vu$, then all of them have to be actively present in $S_{vv_2}$, implying that $\vert S_{vv_2} \vert \ge 3$, a contradiction since $\vert S_{vv_2} \vert =2$. Therefore at least one of the three candidate colors is valid for the edge $vu$.
\end{itemize}

Thus we have a valid color for edge $vu$, a contradiction to the assumption that $G$ is a counter example.
\end{proof}

In view of $Claim$ \ref{clm:clm6}, $Claim$ \ref{clm:clm7} and from $\vert S_{uz} \vert$ , $\vert S_{uw} \vert \le \Delta-1$ and $\vert X \vert = \Delta-2$, it is easy to see that $\vert (S_{uz} \cup S_{uw}) - X \vert \le 2$. Thus recalling that $3,4,5 \notin X$, we infer that $\{3,4,5\} - (S_{uz} \cup S_{uw}) \neq \emptyset$. Now recolor the edge $uz$ using a color $\mu \in  \{3,4,5\} - (S_{uz} \cup S_{uw})$. Clearly $\mu$ is a candidate for the edge $uz$ since $d'(u,w)=2$ and $\mu \notin S_{uz}$. Moreover $\mu$ is valid for $uz$ since if otherwise a $(2,\mu)$ bichromatic cycle has to be formed containing $uw$, implying that $\mu \in S_{uw}$, a contradiction. This reduces the situation to $F_{u} \subseteq  F'_v$, a contradiction to $Claim$ \ref{clm:clm4}.

~~~~~~

\subsubsection{case 2: $\vert F_{v} \cap F_{u} \vert= 1$}

Recall that by $Claim$ \ref{clm:clm3} and $Claim$ \ref{clm:clm2}, $v$ belongs to configurations $B3$, $B4$ or $B5$ and $deg(u)=3$. As before $N(u) =\{v,w,z\}$. Also let $F_{v} \cap F_{u}= \{1\}$.

~~~~~

\begin{clm}
\label{clm:clm8}
With respect to any valid coloring of $G- \{vu\}$, $F_{u} \cap  F'_v = \emptyset$. This implies that $F_{v} \cap F_{u} \subseteq F''_v$.
\end{clm}
\begin{proof}
Suppose not. Then without loss of generality let $c'(v,v_1)=c'(u,z)=1$. Recalling $deg(u)=3$, $\vert F_{u} \vert \le 2$  and thus $\vert F_{u} \cup F_v \vert \le (\Delta-1)+2-1 = \Delta$. It follows that there are at least three candidate colors for the edge $vu$. If none of the candidate colors are valid for the edge $vu$, then all these candidate colors have to be actively present in $S_{vv_1}$, implying that $\vert S_{vv_1} \vert \ge 3$, a contradiction since $\vert S_{vv_1} \vert =2$. It follows that at least one of the three candidate colors is valid for the edge $vu$, a contradiction to the assumption that $G$ is a counter example.
\end{proof}

\noindent In view of $Claim$ \ref{clm:clm8}, $F''(v) \neq \emptyset$ and therefore the vertex $v$ cannot belong to configuration $B5$. We infer that $v$ has to belong to either configuration $B3$ or $B4$. We take care of these two subcases separately below:

~~~~~

\noindent \textbf{subcase 2.1: $v$ belongs to configuration $B3$.} \newline

Since $deg(v) = 5$, we have $\vert F_v \vert = 4$. Let $F_v \cup F_u =\{1,2,3,4,5\}$. By Claim \ref{clm:clm8}, we have $F_v \cap F_{u} = \{1\} \subseteq F''_v = \{c'(v,a),c'(v,b)\}$. Without loss of generality let $c'(u,z)=c'(v,a)=1$. Also let $c'(u,w)=2$, $c'(v,b)=3$, $c'(v,v_1)=4$ and $c'(v,v_2)=5$. Since $\vert F_v \cup F_u \vert =5$, there are $\Delta-2$ candidate colors for the edge $vu$. If none of them are valid then there exists a $(1,\alpha,vu)$ critical path for each $\alpha \in C-(F_v \cup F_{u})=C-\{1,2,3,4,5\}$. Thus  we have the following observation:

\begin{obs}
\label{obs:obs2}
With respect to the coloring $c'$, each color in $C-\{1,2,3,4,5\}$ is actively present in $S_{uz}$ as well as $S_{va}$.
\end{obs}

\begin{clm}
\label{clm:clm9}
$S_{uz} = C-\{1,3,4,5\}$ and $1,4,5 \in S_{uw}$.
\end{clm}
\begin{proof}
Since $C-\{1,2,3,4,5\} \subseteq S_{uz}$ and $\vert S_{uz} - (C-\{1,2,3,4,5\}) \vert \le 1$ we infer that at most one of $4$, $5$ can be present in $S_{uz}$. Suppose one of $4$, $5$ $\in S_{uz}$. Without loss of generality let $4 \in S_{uz}$. Now recolor edge $uz$ using color $5$. It is valid by $Lemma$ \ref{lem:lem1} since $S_{uz} \cap S_{zu} = S_{uz} \cap \{2\} = \emptyset$. Thus we have reduced the situation to $F_u \cap F'_v \neq \emptyset$, a contradiction to $Claim$ \ref{clm:clm8}. Thus we have $4,5 \notin S_{uz}$. Recolor edge $uz$ using color $4$ or $5$. If any one of them is valid then we will have $F_{u} \cap  F'_v \neq \emptyset$ with respect to this new coloring, a contradiction to $Claim$ \ref{clm:clm8}. It follows that none of them are valid. That is, bichromatic cycles get formed due to the recoloring. Clearly the bichromatic cycles have to be $(2,4)$ and $(2,5)$ bichromatic cycles since $c'(u,w)=2$. Thus $2 \in S_{uz}$ and $4,5 \in S_{uw}$. Recalling that $C-\{1,2,3,4,5\} \subseteq S_{uz}$ and $\vert S_{uz} \vert \le \Delta-1$ we can infer that $S_{uz} = C-\{1,3,4,5\}$.

Now if $1 \notin S_{uw}$, then assign color $1$ to edge $uw$ and the color $4$ to edge $uz$. Clearly this recoloring is valid by $Lemma$ \ref{lem:lem1} since $S_{zu} \cap S_{uz} = \{1\} \cap C-\{1,3,4,5\} = \emptyset$. With respect to the new coloring, $F_u \cap F_v=\{1,4\}$ which reduces the situation to $case\ 1$. Thus we infer that $1 \in S_{uw}$. Therefore we have $1,4,5 \in S_{uw}$.
\end{proof}

\begin{clm}
\label{clm:clm10}
$\vert (C-\{1,2,3,4,5\})- S_{uw} \vert \ge 2$.
\end{clm}
\begin{proof}
Since $\vert S_{uw} \vert \le \Delta-1$ there are at least four colors missing from $S_{uw}$. Thus even if colors $2$ and $3$ are missing from $S_{uw}$ there should be at least two colors in $C-\{1,2,3,4,5\}$ that are absent in $S_{uw}$ since $1,4,5 \in S_{uw}$ by $Claim$ \ref{clm:clm9}.
\end{proof}

\noindent Now discard the color on the edge $uw$ to obtain a partial coloring $d$ of $G$ from $c'$.

\begin{clm}
\label{clm:clm11}
With respect to coloring $d$, $\forall \alpha \in C-\{1,3,4,5\}$, there exists a $(1,\alpha,vu)$ critical path.
\end{clm}
\begin{proof}
With respect to the coloring $c'$, there existed $(1,\alpha,vu)$ critical path for all $\alpha \in C-(F_v \cup F_{u})=C-\{1,2,3,4,5\}$ by $Observation$ \ref{obs:obs2}. These critical paths remain unaltered when we get $d$ from $c'$. Thus these critical paths are present in $d$ also. Thus it is enough to prove that there exists $(1,2,vu)$ critical path with respect to the coloring $d$. Let $\theta \in (C-\{1,2,3,4,5\})- S_{uw}$. Note that $\theta$ exists by $Claim$ \ref{clm:clm10}. Now color $\theta$ is a candidate for the edge $uw$ since $\theta \notin S_{uw}$ and $d(u,z)=1$. Recolor the edge $uw$ using color $\theta$ to get a coloring $d'$. The coloring $d'$ is valid since otherwise a $(1,\theta)$ bichromatic cycle has to be created due to the recoloring. This means that there existed a $(1,\theta,uw)$ critical path with respect to coloring $c'$, a contradiction by $Fact$ \ref{fact:fact1} as there already existed a $(1,\theta,vu)$ critical path with respect to the coloring $c'$ by $Observation$ \ref{obs:obs2}. Thus the coloring $d'$ is valid.

Now color $2$ is a candidate for the edge $vu$. If it is valid we get a valid coloring for $G$. Thus it is not valid. This means that there exists a $(1,2,vu)$ critical path with respect to the coloring $d'$ since $F_v\cap F_u=\{1\}$ with respect to the coloring $d'$. Now it is easy to see that this  $(1,2,vu)$ critical path will also exist with respect to coloring $d$. Thus with respect to the coloring $d$, $\forall \alpha \in C-\{1,3,4,5\}$, there exists a $(1,\alpha,vu)$ critical path.
\end{proof}

\begin{obs}
\label{obs:obs3}
Let $Q= (C-\{1,3,4,5\})-S_{uw}$. From Claim \ref{clm:clm10}, we know that $\vert (C-\{1,2,3,4,5\})-S_{uw} \vert \ge 2$. Since $c'(u,w)=2$ we have $2 \notin S_{uw}$. From this we can infer that $2 \in Q$. Thus $\vert Q \vert \ge 3$.
\end{obs}

\begin{clm}
\label{clm:clm12}
There exists a color $\gamma \in Q$ such that $\gamma$ is valid for the edge $vv_1$ or $vv_2$.
\end{clm}
\begin{proof}

Recall that $\vert S_{vv_1} \vert =2$, $\vert S_{vv_2} \vert =2$ and by $Observation$ \ref{obs:obs2}, $\vert Q \vert \ge 3$.

\begin{itemize}
\item If $S_{vv_1} \subset Q$ or $S_{vv_2} \subset Q$. Without loss of generality let $S_{vv_1} \subset Q$. Let $\gamma$ be a color in $Q-S_{vv_1}$. Recolor edge $vv_1$ using color $\gamma$ to get a coloring $d'$. The coloring $d'$ is valid by $Lemma$ \ref{lem:lem1} as $S_{vv_1} \cap S_{v_1v}=\emptyset$ since $Q \cap F_v = \emptyset$.

\item If $S_{vv_1} \nsubseteq Q$ and $S_{vv_2} \nsubseteq Q$. In this case, at most one color in $Q$ can be in $S_{vv_1}$ and the same holds true for $S_{vv_2}$. Thus all the colors of $Q$ except for one are candidates for edge $vv_1$ and all the colors of $Q$ except for one are candidates for edge $vv_2$. Since $\vert Q \vert \ge 3$, we can infer that there exists a color $\gamma \in Q$ which is a candidate for both $vv_1$ and $vv_2$.

~~~~~~

\noindent \textbf{subclaim} \emph{ Color $\gamma$ is valid either for the edge $vv_1$ or for the edge $vv_2$.}
\begin{proof}
Recolor $vv_1$ using color $\gamma$. If $\gamma$ is valid, we are done. If it is not valid, then there has to be a $(\gamma,\theta)$ bichromatic cycle getting formed, where $\theta \in F_v-\{d(v,v_1)\}=F_v-\{4\}=\{1,3,5\}$. But this cannot be a $(\gamma,5)$ bichromatic cycle since $\gamma \notin S_{vv_2}$ (recall that $d(v,v_2)=c'(v,v_2)=5$). Also this cannot be a $(\gamma,1)$ bichromatic cycle since otherwise it implies that there exists a $(1,\gamma,vv_1)$ critical path with respect to the coloring $d$, a contradiction in view of $Fact$ \ref{fact:fact1} as there already exists a $(1,\gamma,vu)$ critical path by $Claim$ \ref{clm:clm11}. Thus it has to be a $(3,\gamma)$ bichromatic cycle, implying that there existed a $(3,\gamma,vv_1)$ critical path with respect to the coloring $d$.

If $\gamma$ is not valid for the edge $vv_1$ we recolor edge $vv_2$ instead, using color $\gamma$ to get a coloring $d'$ form $d$. We claim that the coloring $d'$ is valid. This is because there cannot be a $(\gamma,4)$ bichromatic cycle since $\gamma \notin S_{vv_1}$ (recall that $d(v,v_1)=c'(v,v_1)=4$). Also there cannot be a $(\gamma,1)$ bichromatic cycle since otherwise it implies that there exists a $(1,\gamma,vv_2)$ critical path with respect to the coloring $d$, a contradiction in view of $Fact$ \ref{fact:fact1} as there already exists a $(1,\gamma,vu)$ critical path by $Claim$ \ref{clm:clm11}. Finally there cannot be a $(3,\gamma)$ bichromatic cycle because this implies that there existed a $(3,\gamma,vv_2)$ critical path with respect to the coloring $d$, a contradiction by $Fact$ \ref{fact:fact1} since there already existed a $(3,\gamma,vv_1)$ critical path with respect to the coloring $d$. Thus the coloring $d'$ is valid.
\end{proof}

\end{itemize}

\end{proof}

In view of Claim \ref{clm:clm12}, without loss of generality let $\gamma \in Q$ be valid for the edge $vv_1$. Now we recolor the edge $vv_1$ using color $\gamma$ to get a coloring $d'$.

We claim that none of the colors in $S_{uw}$ were altered in this recoloring. This is because if they are altered then $vv_1$ has to be an edge incident on $w$ and thus one of the end points of $vv_1$ has to be $w$. Since $v$ cannot be $w$, either $v_1$ should be $w$. But we know that $deg(v_1)=3$. Recall that $1,4,5 \in S_{uw}$ and thus $deg(w) \ge 4$. Thus $v_1$ cannot be $w$. Thus none of the colors of $S_{uw}$ are modified while getting $d'$ from $d$. We infer that $\gamma \notin S_{uw}$ since $Q \cap S_{uw} = \emptyset$. Therefore $\gamma$ is a candidate for the edge $uw$ since $d'(u,z)=1$. Now color the edge $uw$ using the color $\gamma$ to get a coloring $d''$. If the coloring $d''$ is valid, then we have $F_{u} \cap F_v = \{1,\gamma\}$. This reduces the situation to $case\ 1$.

On the other hand if the coloring $d''$ is not valid then there has to be a bichromatic cycle formed due to the recoloring of edge $uw$. Since $d''(u,z)=1$, it has to be a $(1,\gamma)$ bichromatic cycle. Recall that there existed a $(1,\gamma,vu)$ critical path with respect to the coloring $d$. Note that to get $d''$ from $d$ we have only recolored two edges namely $vv_1$ and $uw$, both with color $\gamma$. Clearly these recolorings cannot break the $(1,\gamma,vu)$ critical path that existed in $d$, but only can extend it. Thus we can infer that in $d''$ the $(1,\gamma)$ bichromatic cycle passes through $v$ and hence through the edges $va$ and $vv_1$. Now recolor edge $va$ using color $4$ to get a coloring $c$. Recall that $S_{va}= C-\{1,3,4,5\}$ by Claim \ref{clm:clm11} and $S_{av}= F_v-\{c''(v,a)\}=\{3,5,\gamma\}$. Therefore color $4$ is indeed a candidate for edge $va$. Note that by recoloring $va$ using color $4$, we have broken the $(1,\gamma)$ bichromatic cycle that existed in $d''$. Now we claim that the coloring $c$ is valid. Note that $S_{va} \cap S_{av}= S_{va} \cap \{3,5,\gamma\} = \{ \gamma\}$. If a bichromatic cycle gets formed due to this recoloring then it has to be $(4,\gamma)$ bichromatic cycle, implying that $4 \in S_{vv_1}$. But $S_{vv_1}(c)=S_{vv_1}(d'')=S_{vv_1}(d)$ and $4 \notin S_{vv_1}(d)$ since $d(v,v_1)=4$. Thus $4 \notin S_{vv_1}(c)$, a contradiction. Thus the coloring $c$ is valid. With respect to the coloring $c$, we have $F_v \cap F_u = \{\gamma\} \subset F'_v$, a contradiction to $Claim$ \ref{clm:clm8}.

~~~~~~~

\noindent \textbf{subcase 2.2: $v$ belongs to configuration $B4$.} \newline

We have $deg(v)=6$ and therefore $\vert F_v \vert = 5$. Moreover $\vert F''_v \vert = 1$ and $\vert F'_v \vert = 4$. By $Claim$ \ref{clm:clm8}, $F_v \cap F_u =\{1\} \subseteq F''_v$. Without loss of generality let $c'(u,z)=c'(v,a)=1$. Also let $c(u,w)=2$, $F'_v=\{3,4,5,6\}$ and $Z= \{3,4,5,6\}$. There are $\Delta-3$ candidate colors for the edge $vu$. If none of them are valid then there exist $(1,\alpha,vu)$ critical path for each $\alpha \in C-(F_v \cup F_{u})=C-\{1,2,3,4,5,6\}$. Thus we have the following observation:

\begin{obs}
\label{obs:obs4}
With respect to the coloring $c'$, each color in $C-\{1,2,3,4,5,6\}$ is actively present in $S_{uz}$ as well as $S_{va}$.
\end{obs}

\begin{clm}
\label{clm:clm13}
$S_{uz} \supseteq C-\{1,3,4,5,6\}$ and $1 \in S_{uw}$. Also at least three of the colors from $Z$ are present in $S_{uw}$.
\end{clm}
\begin{proof}
As we have seen above $C-\{1,2,3,4,5,6\} \subseteq S_{uz}$. Suppose $2 \notin S_{uz}$. Note that every color in $C-(S_{uz} \cup S_{zu})$ is a candidate for $uz$. Now $S_{zu}=\{c'(u,w)\}=\{2\}$. Moreover $\vert S_{uz} \vert \le \Delta-1$ and thus $S_{uz}$ can have at most two more colors other than those in $C-\{1,2,3,4,5,6\}$. From this we can infer that at least two of the colors in $Z$ are candidates for the edge $uz$. They are also valid by $Lemma$ \ref{lem:lem1} since $S_{uz} \cap S_{zu}= S_{uz} \cap \{2\} = \emptyset$. Thus we can reduce the situation to $F_u \cap F'_v \neq \emptyset$, by assigning one of the valid colors from $Z$ to $uz$, thereby getting a contradiction to $Claim$ \ref{clm:clm8}. Thus we infer that $2 \in S_{uz}$. Therefore we get $S_{uz} \supseteq C-\{1,3,4,5,6\}$. Since $\vert S_{uz} \vert \le \Delta-1$ and $\vert C-\{1,3,4,5,6\} \vert = \Delta-2$ we can infer that $\vert Z \cap S_{uz} \vert \le 1$.

If any one of the colors in $Z - S_{uz}$ is valid for the edge $uz$, then it will reduce the situation to $F_u \cap F'_v \neq \emptyset$, a contradiction to $Claim$ \ref{clm:clm8}. Thus none of these colors are valid for the edge $uz$. Therefore there should be bichromatic cycles getting formed when we try to recolor edge $uz$ using any of these colors. These bichromatic cycles have to be $(2,\mu)$ bichromatic cycles for each color $\mu \in Z - S_{uz}$ since $c'(u,w)=2$. Thus we can infer that at least three of the colors from $Z$ are present in $S_{uw}$ since $\vert Z - S_{uz} \vert \ge 4 -1 = 3$.

Now if $1 \notin S_{uw}$, then assign color $1$ to edge $uw$ and a color $\mu \in Z - S_{uz}$ to edge $uz$. Clearly this recoloring is valid by $Lemma$ \ref{lem:lem1} since $S_{zu} \cap S_{uz} = \{1\} \cap S_{uz} = \emptyset$ ($1 \notin S_{uz}$ since $c'(u,z)=1$). With respect to the new coloring, $F_u \cap F_v=\{1,\mu\}$ which reduces the situation to $case\ 1$. Thus we infer that $1 \in S_{uw}$.
\end{proof}

\begin{clm}
\label{clm:clm14}
$\vert (C-\{1,2,3,4,5,6\})- S_{uw} \vert \ge 2$.
\end{clm}
\begin{proof}
Since $\vert S_{uw} \vert \le \Delta-1$, we have $\vert C-S_{uw} \vert \ge 4$. Now since $\vert Z \cap S_{uw} \vert \ge 3$ and $1 \in S_{uw}$, $\vert \{1,2,3,4,5,6\} \cap S_{uw} \vert \ge 4$. It follows that $\vert (C - S_{uw}) \cap \{1,2,3,4,5,6\} \vert \le 2$ and the Claim follows.
\end{proof}

\noindent Now discard the color on the edge $uw$ to obtain a partial coloring $d$ of $G$ from $c'$.

\begin{clm}
\label{clm:clm15}
With respect to coloring $d$, $\forall \alpha \in C-\{1,3,4,5,6\}$, there exists a $(1,\alpha,vu)$ critical path and thus $C-\{1,3,4,5,6\} \subseteq S_{va}$.
\end{clm}
\begin{proof}
With respect to the coloring $c'$, there existed a $(1,\alpha,vu)$ critical path for each $\alpha \in C-(F_v \cup F_{u})=C-\{1,2,3,4,5,6\}$ by $Observation$ \ref{obs:obs4}. These critical paths remain unaltered when we get $d$ from $c'$. Thus these critical paths are present in $d$ also. Thus it is enough to prove that there exists a $(1,2,vu)$ critical path with respect to the coloring $d$. Let $\theta \in (C-\{1,2,3,4,5,6\})- S_{uw}$. Note that $\theta$ exists by $Claim$ \ref{clm:clm14}. Now color $\theta$ is a candidate for the edge $uw$ since $\theta \notin S_{uw}$ and $d(u,z)=1$. Recolor the edge $uw$ using color $\theta$ to get a coloring $d'$. The coloring $d'$ is valid since otherwise a $(1,\theta)$ bichromatic cycle has to be created due to the recoloring. This means that there existed a $(1,\theta,uw)$ critical path with respect to coloring $c'$, a contradiction by $Fact$ \ref{fact:fact1} as there already existed a $(1,\theta,vu)$ critical path with respect to the coloring $c'$ by $Observation$ \ref{obs:obs4}. Thus the coloring $d'$ is valid.

Now color $2$ is a candidate for the edge $vu$. If it is valid we get a valid coloring for $G$. Thus it is not valid. This means that there exists a $(1,2,vu)$ critical path with respect to the coloring $d'$ since $F_v\cap F_u=\{1\}$ with respect to the coloring $d'$. Now it is easy to see that this  $(1,2,vu)$ critical path will also exist with respect to coloring $d$. Thus with respect to the coloring $d$, $\forall \alpha \in C-\{1,3,4,5,6\}$, there exists a $(1,\alpha,vu)$ critical path.
\end{proof}

\begin{obs}
\label{obs:obs5}
Let $Q= (C-\{1,3,4,5,6\})-S_{uw}$. From Claim \ref{clm:clm14}, we know that $\vert (C-\{1,2,3,4,5,6\})-S_{uw} \vert \ge 2$. Since $c'(u,w)=2$ we have $2 \notin S_{uw}$. From this we can infer that $2 \in Q$. Thus $\vert Q \vert \ge 3$.
\end{obs}

Recall that $\vert S_{vv_i} \vert =2$, for $i \in \{1,2,3,4\}$ and by $Observation$ \ref{obs:obs5}, $\vert Q \vert \ge 3$. We know that $S_{va} \supseteq C-\{1,3,4,5,6\}$ by $Claim$ \ref{clm:clm15}. Since $\vert C-\{1,3,4,5,6\} \vert = \Delta-2$ and $\vert S_{va} \vert \le \Delta-1$ we have $\vert Z \cap S_{va} \vert = \vert \{3,4,5,6 \} \cap S_{va} \vert \le 1$. We make the following assumption:

\begin{asm}
\label{asm:asm2}
If $Z \cap S_{va} \neq \emptyset$, let $\{\alpha\} = Z \cap S_{va}$ and let $d(v,v_t)=\alpha$, where $t \in \{1,2,3,4\}$. Let $\beta \in (Z-\{\alpha\})-S_{vv_t}$. If $Z \cap S_{va} = \emptyset$, then let $\beta$ be any color in $Z$.
\end{asm}

We now plan to recolor one of the edges in $\{vv_1,vv_2,vv_3,vv_4\}$ using a specially selected color  $\gamma \in Q$. After this we will also use the same color $\gamma$ to recolor edge $uw$, with the intention of reducing the situation to $case\ 1$. Below we give the recoloring procedure for the rest of the proof starting from the current coloring $d$ in $3$ steps. The final coloring $c$ of $G- \{vu\}$ that we obtain at the end of $Step3$  will give the required contradiction.

~~~~~~

\noindent \textbf{Step1: With respect to the coloring $d$,}
\begin{itemize}
\item[(i)] \textbf{If one of the edges $vv_i$, for $i \in \{1,2,3,4\}$ is such that  $S_{vv_i} \subset Q$, then recolor that edge with any color $\gamma \in Q-S_{vv_i}$. We call the edge that we chose to recolor as $(v,v_{t'})$.}

\item[(ii)] \textbf{If $\forall i \in \{1,2,3,4\}$, $S_{vv_i} \nsubseteq Q$, then we select an edge $vv_{t'}$, where $t' \in \{1,2,3,4\}$ such that $d(v,v_{t'})=\beta$ (See $Assumption$ \ref{asm:asm2}). Now recolor the edge $vv_{t'}$ with a suitably selected (see the proof of $Claim$ \ref{clm:clm16}) color in $Q-S_{vv_{t'}}$.}
\end{itemize}
\textbf{The resulting coloring after performing $Step1$ is named $d'$.}

~~~~~~

\begin{clm}
\label{clm:clm16}
There exists a color $\gamma \in Q$ such that the coloring $d'$ obtained after $Step1$ is valid.
\end{clm}
\begin{proof}
At the beginning of $Step1$, we had the following possible cases:

\begin{itemize}
\item[(i)] \textbf{ One of the edges $vv_i$, for $i \in \{1,2,3,4\}$ is such that  $S_{vv_i} \subset Q$:} \newline
Let $\gamma$ be a color in $Q-S_{vv_i}$. Recolor edge $vv_i$ using color $\gamma$ to get a coloring $d'$. The coloring $d'$ is valid by $Lemma$ \ref{lem:lem1} as $S_{vv_i} \cap S_{v_iv}=\emptyset$ since $Q \cap F_v = \emptyset$.

\item[(ii)] \textbf{ $S_{vv_i} \nsubseteq Q$, for each $i \in \{1,2,3,4\}$:} \newline
Let $t'$ be as defined in $Step1$. Clearly all the colors in $Q-S_{vv_{t'}}$ are candidates for $vv_{t'}$ since $Q \cap F_v = \emptyset$. Note that since $S_{vv_i} \nsubseteq Q$ we have $\vert Q \cap S_{vv_{t'}} \vert \le 1$ and therefore $\vert Q-S_{vv_{t'}} \vert \ge 2$. If any one of the candidate colors is valid for the edge $vv_{t'}$, the statement of the Claim is obviously true. On the other hand if none of them are valid, then there has to be a $(\gamma,\theta)$ bichromatic cycle getting formed, for some $\theta \in F_v-\{d(v,v_{t'})\}= F_v -\{\beta\}$ when we try to recolor edge $vv_{t'}$ using color $\gamma$, for each $\gamma \in Q-S_{vv_{t'}}$. Note that $\theta \neq 1$ because if a $(\gamma,1)$ bichromatic cycle gets formed, then there has to be a $(1,\gamma,vv_{t'})$ critical path with respect to the coloring $d$, a contradiction in view of $Fact$ \ref{fact:fact1} as there already exists a $(1,\gamma,vu)$ critical path by $Claim$ \ref{clm:clm15}. Thus $\theta \in F'_v-\{d(v,v_{t'})\}$ since $F''_v=\{1\}$. Therefore we have $\vert (F'_v-\{d(v,v_{t'})\}) \cap S_{vv_{t'}} \vert \ge 1$. We have the following cases:

\begin{itemize}
\item \textbf{$\vert (F'_v-\{d(v,v_{t'})\}) \cap S_{vv_{t'}} \vert = 1$:}
Let $S_{vv_{t'}} \cap (F'_v-\{d(v,v_{t'})\})=d(v,v')$, for $v' \in \{v_1,v_2,v_3,v_4\}-\{v_{t'}\}$. Thus all the candidate colors of $vv_{t'}$, namely all the colors of $Q-S_{vv_{t'}}$ should form bichromatic cycles passing through the edge $vv'$, implying that $Q-S_{vv_{t'}} \subset S_{vv'}$. But $\vert Q-S_{vv_{t'}} \vert \ge 2$ and $\vert S_{vv'} \vert =2$. Thus $S_{vv'} = Q-S_{vv_{t'}} \subseteq Q$, a contradiction.

\item \textbf{$\vert (F'_v-\{d(v,v_{t'})\}) \cap S_{vv_{t'}} \vert = 2$:}  This means that $S_{vv_{t'}} \subseteq F'_v$ and therefore we have $Q \cap S_{vv_{t'}} = \emptyset$. Thus $\vert Q- S_{vv_{t'}} \vert = \vert Q \vert \ge 3$. Therefore there are at least three candidate colors for the edge $vv_{t'}$. Let $S_{vv_{t'}} \cap (F'_v-\{d(v,v_{t'})\})=\{d(v,v'),d(v,v'')\}$, for $v',v'' \in \{v_1,v_2,v_3,v_4\}-\{v_{t'}\}$. Since for each candidate color we have a bichromatic cycle, we can infer that there are at least three bichromatic cycles, each of them passing through either $vv'$ or $vv''$. Thus at least two bichromatic cycles have to pass through one of $vv'$ and $vv''$. But since $\vert S_{vv'} \vert =2$ and $\vert S_{vv''} \vert =2$, we can infer that either  $S_{vv'} \subseteq Q$ or $S_{vv''} \subseteq Q$, a contradiction.
\end{itemize}

\end{itemize}

\end{proof}

~~~~~

\noindent \textbf{Step2: Let $\gamma$ be the color which was used to recolor the edge $vv_{t'}$ in $Step1$. Now recolor edge $uw$ with color $\gamma$ to get a coloring $d''$.}

~~~~~~

\begin{clm}
\label{clm:clm17}
The coloring $d''$ is proper.
\end{clm}
\begin{proof}
We claim that none of the colors in $S_{uw}$ were altered in $Step1$. This is because if they are altered then the edge $vv_{t'}$ should be incident on $w$ and thus one of the end points of $vv_{t'}$, where $t' \in \{1,2,3,4\}$, has to be $w$. Since $v$ cannot be $w$, $v_{t'}$ should be $w$. But we know that $deg(v_i)=3$. Recall that $\vert Z \cap S_{uw} \vert \ge 3$ by $Claim$ \ref{clm:clm13} and thus $\vert S_{uw} \vert \ge 3$. Therefore $deg(w) \ge 4$. Thus $v_{t'}$ cannot be $w$. Thus none of the colors of $S_{uw}$ are modified while getting $d'$ from $d$. Recall that $Q= (C-\{1,3,4,5,6\})-S_{uw}$ and thus $\gamma \notin S_{uw}$. Therefore $\gamma$ is a candidate for the edge $uw$ since $d(u,z)=1$. Thus the coloring $d''$ is proper.
\end{proof}

If the coloring $d''$ is valid, then we have $F_{u} \cap F_v = \{1,\gamma\}$ for a valid coloring of $G-\{vu\}$. This reduces the situation to $case\ 1$. Thus coloring $d''$ is not valid. Since the coloring $d''$ is not valid, there has to be a bichromatic cycle formed due to the recoloring of edge $uw$. Since $d''(u,z)=1$, it has to be a $(1,\gamma)$ bichromatic cycle. Recall that there existed a $(1,\gamma,vu)$ critical path with respect to the coloring $d$ by $Claim$ \ref{clm:clm15}. Note that to get $d''$ from $d$ we have only recolored two edges namely $vv_{t'}$ and $uw$, both with color $\gamma$. Clearly these recolorings cannot break the $(1,\gamma,vu)$ critical path that existed in $d$, but can only extend it. Thus we can infer that in $d''$ the $(1,\gamma)$ bichromatic cycle passes through $v$ and hence through the edges $va$ and $vv_{t'}$. Also note that this can happen only when we have $1 \in S_{vv_{t'}}$. Thus $S_{vv_{t'}} \nsubseteq Q$. It means that substep $(ii)$ of $Step1$ was executed; and the color on $vv_{t'}$ with respect to coloring $d$ was $\beta$ (from $Assumption$ \ref{asm:asm2}). We break the $(1,\gamma)$ bichromatic cycle as follows:

~~~~~~~

\noindent \textbf{Step3: Recolor the edge $va$ with color $\beta$ (see in $Assumption$ \ref{asm:asm2}) to get a coloring $c$}.

~~~~~~~~~~~~~

\begin{clm}
\label{clm:clm18}
The coloring $c$ is valid.
\end{clm}
\begin{proof}
Recall by $Assumption$ \ref{asm:asm2} that $\beta \notin S_{va}$. Also clearly $\beta \notin F_v(d'')$ since we recolored $vv_{t'}$ by a color $\gamma \in Q$ to get $d''$ form $d$ ($\beta \neq \gamma$ since $\beta \in F_v(d)$ and $F_v(d) \cap Q = \emptyset$). Therefore color $\beta$ is a candidate for edge $va$. Note that by recoloring $va$ using color $\beta$, we have broken the $(1,\gamma)$ bichromatic cycle that existed in $d''$. We claim that the coloring $c$ is valid. Otherwise there has to be a bichromatic cycle involving $\beta$ and a color in $S_{va} \cap S_{av}$. But $S_{av} = (Z-\{\beta\}) \cup \{\gamma\} = (\{3,4,5,6\}- \{\beta\}) \cup \{\gamma\}$. Since  with respect to $d''$ there was a $(1,\gamma)$ bichromatic cycle passing through the edges $va$ and $d''(v,a)=1$, we have $\gamma \in S_{va} \cap S_{av}$. But there cannot be a $(\beta,\gamma)$ bichromatic cycle getting formed in $c$ since such a cycle should contain edge $vv_{t'}$ and thus $\beta \in  S_{vv_{t'}}$. But $S_{vv_{t'}}(c)=S_{vv_{t'}}(d'')=S_{vv_{t'}}(d)$ and $\beta \notin S_{vv_{t'}}(d)$ since $d(v,v_{t'})=\beta$. Thus $\beta \notin S_{vv_1}(c)$, a contradiction. Thus there cannot be a $(\beta,\gamma)$ bichromatic cycle.

Thus if the coloring $c$ is not valid then there has to be a bichromatic cycle involving $\beta$ and one of the colors in $Z-\{\beta\} \cap S_{va}$. We know by $Assumption$ \ref{asm:asm2} that $Z \cap S_{va} = \alpha$. Thus it has to be a $(\beta,\alpha)$ bichromatic cycle. Since $c(v,v_t)=d(v,v_t)=\alpha$, this bichromatic cycle contains the edge $vv_t$ and hence $\beta \in S_{vv_t}$, a contradiction to the way $\beta$ was selected in $Assumption$ \ref{asm:asm2}. Thus there cannot be a $(\beta,\alpha)$ bichromatic cycle. Thus the coloring $c$ is valid.
\end{proof}

\noindent With respect to the coloring $c$, we have $F_v \cap F_u = \{\gamma\} \subset F'_v$, a contradiction to $Claim$ \ref{clm:clm8}.

\subsection{There exists no vertex $v$ that belongs to one of the configurations $B2$, $B3$, $B4$ or $B5$}

This means that there exists a vertex $v$ that belongs to configuration $B1$, i.e., $deg(v)=2$. Let $Q = \{u \in V\ :\ deg(u) =2\}$. First we claim that $Q$ is an independent set in $G$. Otherwise let $u',u \in Q$ be such that $(u,u') \in E(G)$. Now since $G$ is a minimum counter example, $G-\{uu'\}$ is acyclically edge colorable using $\Delta+3$ colors. Let $c'$ be a valid coloring of $G-\{uu'\}$. Now if $F_u \cap F_{u'} = \emptyset$, then there are $\Delta+3 - 2= \Delta+1$, candidate colors for the edge $uu'$. Since $S_{uu'} \cap S_{u'u} = \emptyset$, by $Lemma$ \ref{lem:lem1}, all the candidate colors are valid for the edge $uu'$. On the other hand if $\vert F_u \cap F_{u'} \vert = 1$, then there are $\Delta+3 - 1= \Delta+2$ candidate colors for the edge $uu'$. Let $N(u)=\{u',u''\}$. If none of them are valid then all those candidate colors have to be actively present in $S_{uu''}$, implying that $\vert S_{uu''} \vert \ge \Delta+2$, a contradiction since $\vert S_{uu''} \vert \le \Delta-1$. Thus there exists a valid coloring of $G$, a contradiction to the assumption that $G$ is a counter example. We infer that $Q$ is an independent set in $G$.

Now delete all the vertices in $Q$ from $G$ to get a graph $G'$. Clearly the graph $G'$ has at most $2\vert V(G') \vert -1$ edges since $Q$ is an independent set. It follows by $Lemma$ \ref{lem:lem4} that there should be a vertex $v'$ in $G'$ such that $v'$ is the pivot of one of the configurations $B1-B5$, say $B'=\{v'\} \cup N_{G'}(v')$. But with respect to graph $G$, $\{v'\} \cup N_{G'}(v')$ did not form any of the configurations $B1-B5$. This means that the degree of at least one of the vertices in $\{v'\} \cup N_{G'}(v')$ should have got decreased by the removal of $Q$ from $G$. Let $P$ be the set of vertices in $\{v'\} \cup N_{G'}(v')$ whose degrees got reduced due to the removal of $Q$ from $G$, i.e., $P=\{z \in \{v'\} \cup N_{G'}(v'): deg_{G'}(z) < deg_{G}(z)\}$.

For a vertex $x \in V(G)$, let $M''_{G}(x)=\{u \in N_{G}(x) : deg_{G}(u) > 3\}$ and $M'_{G}(x)=N_{G}(x) - M''_{G}(x)$. Note that in all the configurations defined in $Lemma$ \ref{lem:lem4}, the main criteria which characterizes each configuration is the degree of the pivot $v'$ and the degrees of the vertices in $N'(v')$. We make the following claim:

\begin{clm}
\label{clm:clm19}
There exists a vertex $x$ in $P$ such that $\vert M''_{G}(x) \vert \le 3$.
\end{clm}
\begin{proof}
It is easy to see that $M''_{G}(x) \subseteq N_{G'}(x)$. If there exists a vertex in $P$, whose degree is at most $3$, say $x$, then we have $\vert M''_{G}(x) \vert \le 3$. Thus we can assume that the degree of any vertex in $P$ is at least $4$.

Now suppose the pivot vertex $v'$ is in $P$. Then let $x=v'$. It is clear that $v'$ has to be in one of the configuration $B3-B5$. In any of these configurations there can be at most two neighbours with degree greater than $3$. Note that in this case all the degree 3 neighbours of $x=v'$ in $G'$ are of degree $3$ in $G$ also since otherwise $P$ will contain a vertex of degree at most $3$, a contradiction. Thus we have $\vert M''_{G}(x) \vert \le 2$.

The only remaining case is when $v' \notin P$. Since the degree of $v'$ has not changed and $\{v'\} \cup N_{G}(v')$ was not in any configuration in $G$, it means that one of the vertex in $N'(v')$ has had its degree decreased. We call that vertex as $x$. Since the degree of any vertex in $P$ is at least $4$, $deg_{G'}(x) \ge 4$. Since we can have degree $\ge 4$ vertex in $N'(v')$ only if $\{v'\} \cup N_{G}(v')$ forms a configuration $B2$, we infer that $deg_{G'}(x) =4$. Moreover $deg_{G'}(v')=deg_{G}(v') =3$. Thus we have $\vert M''_{G}(x) \vert \le \vert N_{G'}(x)-\{v'\} \vert \le 4 - 1 = 3$.

Thus we have $\vert M''_{G}(x) \vert \le 3$.

\end{proof}

In $G$, let $y$ be a two degree neighbour of vertex $x$ - selected in $Claim$ \ref{clm:clm19} - such that $N(y)=\{x,y'\}$. Now by induction $G-\{xy\}$ is acyclically edge colorable using $\Delta+3$ colors. Let $c'$ be a valid coloring of $G-\{xy\}$. With respect to the coloring $c'$ let $F'_{x}(c')=\{c'(x,z) \vert z \in M'(x)\}$ and $F''_{x}(c')=\{c'(x,z) \vert z \in M''(x)\}$ i.e., $F''_{x}=F_x-F'_x$.

Now if $c'(y,y') \notin F_x$ we are done as there are at least three candidate colors which are also valid by $Lemma$ \ref{lem:lem1}. We know by $Claim$ \ref{clm:clm19} that $\vert F''_x \vert \le 3$. If $c'(y,y') \in F'_x$, then let $c=c'$. Else if $c'(y,y') \in F''_x$, then recolor edge $yy'$ using a color from $C-(S_{yy'} \cup F''_x)$ to get a coloring $c$ (Note that $\vert C-(S_{yy'} \cup F''_x) \vert \ge \Delta+3-(\Delta-1+3) = 1$ and since $y$ is a pendant vertex in $G-\{xy\}$ the recoloring is valid). Now if $c(y,y') \notin F_x$ the proof is already discussed. Thus $c(y,y') \in F'_x$.

With respect to coloring $c$, let $a \in M'(x)$ be such that $c(x,a)=c(y,y')=1$. Now if none of the candidate colors in $C-(F_x \cup F_{y})$ are valid for the edge $xy$, then all those candidate colors have to be actively present in $S_{xa}$, implying that $\vert S_{xa} \vert \ge \vert C- (F_x \cup F_{y}) \vert \ge \Delta+3 - (\Delta-1+1-1) = 4$, a contradiction since $\vert S_{xa} \vert \le 2$ (Recall that $a \in M'(x)$ and $deg(a) \le 3$). Thus we have a valid color for the edge $xy$, a contradiction to the assumption that $G$ is a counter example.

\end{proof}


\end{document}